\documentclass[10pt,conference,letterpaper]{IEEEtran}
\usepackage[letterpaper, left=54pt, right=54pt, top = 54pt, bottom=54pt]{geometry}
\usepackage{hyperref}
\usepackage{amsthm}
\usepackage{epsfig}
\usepackage{xfrac}
\usepackage{amsmath}
\usepackage{cases}
\usepackage{amssymb,mathrsfs,dsfont}
\usepackage{enumerate}
\usepackage{enumitem}
\usepackage{color}
\usepackage{soul}
\usepackage{textcomp}
\usepackage{mathtools}
\usepackage{scrextend}
\usepackage{stackrel}
\usepackage{dsfont}
\usepackage{xparse}
\usepackage{tikz,pgfplots}
\usepackage{kbordermatrix}
\usepackage[utf8]{inputenc}
\usepackage{xfrac}
\newtheorem{theorem}{Theorem}
\newtheorem{lemma}{Lemma}

\newtheorem{corollary}{Corollary}

\newcommand{\bs}[1]{\boldsymbol{#1}}
\newcommand{\mybf}[1]{{\bf #1}}

\newcommand\norm[1]{\left\lVert#1\right\rVert}

\def\md{\mathbb}
\def\eps{\varepsilon}

\def\Expt{\md{E}}

\def\tn{\textnormal}

\newcommand{\dfn}{\stackrel{\tn{def}}{=}}

\def\p2p{point-to-point}

\newcommand{\indfunc}[1]{\mathds{1}\left(#1\right)}
\newcommand{\boldone}{\mathds{1}}
\newcommand{\minhalf}[1]{<\hspace{-1mm}#1\hspace{-1mm}>}

\IEEEoverridecommandlockouts
\renewcommand{\baselinestretch}{0.868}
\mathtoolsset{showonlyrefs=true}
\usepackage[noadjust]{cite}

\NewDocumentCommand{\ftheta}{mg}{f\left(#1\mid\IfNoValueTF{#2}{}{#2,}\Theta=+1\right)}

\begin{document}
\title{\vspace{18pt} 
	Interactive Coding for Markovian Protocols}
\author{Assaf Ben-Yishai, Ofer Shayevitz and Young-Han Kim 
	\thanks{A. Ben-Yishai and O. Shayevitz are with the Department of EE--Systems, Tel Aviv University, Tel Aviv, Israel. 
		Y.-H.~Kim is with the Department of Electrical and Computer Engineering, University of California, San Diego, La Jolla, CA 92093 USA.
		Emails: \{assafbster@gmail.com, ofersha@eng.tau.ac.il, yhk@ucsd.edu\}. The work of A. Ben-Yishai was partially supported by an ISF grant no. 1367/14. The work of O. Shayevitz was supported by an ERC grant no. 639573, a CIG grant no. 631983, and an ISF grant no. 1367/14.}}
\maketitle

\begin{abstract}
We address the problem of simulating an arbitrary Markovian interactive protocol over binary symmetric channels with crossover probability $\varepsilon$. We are interested in the achievable rates of reliable simulation, i.e., in characterizing the smallest possible blowup in communications such that a vanishing error probability (in the protocol length) can be attained. Whereas for general interactive protocols the output of each party may depend on \textit{all} previous outputs of its counterpart, in a (first order) Markovian protocol this dependence is limited to the last observed output only. In the special case where there is no dependence on previous outputs (no interaction), the maximal achievable rate is given by the (one-way) Shannon capacity $1-h(\varepsilon)$.  For Markovian protocols, we first show that a rate of $\frac{2}{3}(1-h(\varepsilon))$ can be trivially achieved. We then describe a more involved coding scheme and provide a closed-form lower bound for its rate at any noise level $\varepsilon$. Specifically, we show that this scheme outperforms the trivial one for any $\varepsilon<0.044$, and achieves a rate higher than $\frac{1-h(\varepsilon)}{1+h(\varepsilon)+h\left(<\varepsilon(2-\varepsilon)>\right)}=1-\Theta(h(\varepsilon))$ as $\varepsilon\to 0$, which is order-wise the best possible. This should be juxtaposed with a result of Kol and Raz that shows the capacity for interactive protocols with alternating rounds is lower bounded by $1-O(\sqrt{h(\varepsilon)})$. 
\end{abstract}

\section{Introduction} 
Suppose Alice and Bob would like to communicate using some interactive communication protocol, where at time point $i$ Alice sends the bit $X_i^A$ and Bob then replies with the bit $X_i^B$ (after having observed Alice's transmission). The transcript associated with their protocol is therefore 
\begin{align}
X_1^A,X_1^B,X_2^A,X_2^B,\cdots,X_n^A,X_n^B.
\end{align}
where 
\begin{align}\label{eq:FullInteraction}
X_i^A = f_i^A\left(\mybf{X}_1^{i-1,B}\right);\quad X_i^B =f_i^B\left(\mybf{X}_1^{i,A}\right). 
\end{align}
The \textit{transmission functions} $f_i^A(\cdot)$, $f_i^B(\cdot)$ depend on the time index $i$ and the identity of the speaker (Alice or Bob) and are unknown to the other party. In general, these functions may depend on the entire set of past inputs observed by either Alice or Bob, i.e. $\mybf{X}_1^{i-1,B}$ or $\mybf{X}_1^{i,A}$ respectively. We refer to the transcript $X_1^A,X_1^B,X_2^A,X_2^B,\cdots,X_n^A,X_n^B$ as the \textit{clean transcript}, where "clean" is used to indicate that Alice and Bob receive their counterpart's transmission without any noise.

Suppose now that Alice and Bob are connected through two independent binary symmetric channels (BSCs) with parameter $\eps$. Namely, Alice receives Bob's transmission with additive noise: $Y_i^B = X_i^B + Z_i^B$, and Bob received Alice's transmission with additive noise $Y_i^A = X_i^A + Z_i^A$, where $\{Z_i^A,Z_i^B\}$ are mutually independent Bernoulli i.i.d. sequence with $\Pr(Z_i^A=1)=\Pr(Z_i^B=1)=\eps$ and "$+$" is addition over $\mathbb{GF}(2)$. Alice and Bob would like to devise a coding scheme that would allow them to reliably simulate the clean transcript over the noisy BSCs. Reliable simulation in this context means that for any Markovian protocol, the probability of either Alice or Bob making an error in recovering the clean transcript goes to zero with the transcript length. To that end, they will need to exchange a larger number of bits; the communication rate of their coding scheme is hence defined to be the total number of bits in the clean transcript divided by the total number of channel uses consumed by their scheme. As usual, one is interested in characterizing the \textit{capacity}, namely the maximal rate for which reliable simulation is possible.   

The problem described above was originally introduced and studied by Schulman \cite{schulman1996coding}. In this seminal work, he showed that reliable simulation with a positive rate (i.e., a positive capacity) can be achieved for any $\eps\neq 1/2$. Kol and Raz \cite{kol2013interactive} further studied the problem in the limit of $\eps\to 0$ and introduced a scheme achieving a rate of $1-O(\sqrt{h(\eps)})$ (where $h(\cdot)$ denotes the binary entropy function). They also showed that for a larger class of protocols with non-alternating rounds the rate is upper bounded by $1-\Omega(\sqrt{h(\eps)})$.

 This demonstrated a separation between one-way and interactive communications, as the one-way capacity is given by $1-h(\eps)$. In \cite{haeupler2014interactive}, Haeupler 
examined a more flexible channel model than ours, in which at every time slot Alice and Bob can independently decide if they want to use the channel as a transmitter or as a receiver. This flexibility can potentially lead to collisions, but was shown to eventually increase the achievable rate to $1-O(\sqrt{\eps})$. Haeupler also conjectured that this rate is order-wise tight under adaptive transmission order, i.e., that the rate of any such reliable scheme is upper bounded by $1-\Omega(\sqrt{\eps})$. We note that the general problem of exactly determining the capacity for any fixed $\eps$ in the interactive setup is still wide open. 

In order to better understand the gap between the one-way and interactive setups for $\eps\to 0$, Haeupler and Velingker \cite{haeupler2017bridging} considered a more restrictive family of protocols that are ``less interactive'', where Alice and Bob have some limited average lookahead, i.e., can often speak for a while without requiring further input from their counterpart (hence, can use short error correcting codes). They showed (also for adversarial noise) that when this average lookahead is $\textrm{poly}(1/\eps)$ then the capacity is $1-O(h(\eps))$, i.e., is order-wise the same as the one-way capacity. 

In this work, rather than restricting the ``interactiveness'' of the protocol as above, we restrict the \textit{memory} of the protocol. Specifically, we consider \textit{Markovian protocols} for which the lookahead can be as short as $1$ (highly interactive), but where Alice and Bob need only recall the last bit they have received. 
For these Markovian Protocols, we provide lower bounds for the capacity for all values of $\eps$, and not only in the limit $\eps\to 0$. 

\subsection{Markovian Protocols}
A (first order) Markovian protocol is a protocol in which each party needs to know only the last transmission of its counterpart in order to decide what to send next, and not the entire set of past transmissions. Namely, 
\begin{align}
X_i^A =f_i^A(X_{i-1}^B);\quad X_i^B =f_i^B(X_{i}^A).
\end{align}
where now, in contrast to \eqref{eq:FullInteraction}, the transmission functions  $f_i^A(\cdot)$, $f_i^B(\cdot)$ depend only on what was last received ($X_{i-1}^B$ and $X_{i}^A$ respectively). It is important to note that the non-interactive communication problem is a special case where $f_i^A(\cdot)$, $f_i^B(\cdot)$ are a sequence of constant valued functions that do not depend on the output of the second party.  

The \textit{rate} of any communication scheme that attempts to simulate the clean transcript is defined by
\begin{align}
R = \frac{2n}{\tilde{n}}
\end{align}
where $2n$ is the length of the clean transcript, and $\tilde{n}$ is the number of channel uses required by the scheme. 

The probability of error attained by a scheme is defined to be the maximal probability that either Alice or Bob fail to exactly simulate the clean transcript, where the maximum is taken over all possible Markovian protocols. A sequence of schemes with rate at least $R$ and error probability approaching zero is said to achieve the rate $R$. The capacity for Markovian protocols over BSCs is the supremum over all such achievable rates, and is denoted by $C_{\textrm{Markov}}(\eps)$. Note that $C_{\textrm{Markov}}(\eps)$ cannot exceed the one-way Shannon capacity of the BSC, i.e.,  
\begin{align}\label{eq:shannon_capacity}
  C_{\textrm{Markov}}(\eps) \leq 1-h(\eps), 
\end{align}
as this is the maximal achievable rate for the special case of non-interactive protocols. Below we derive lower bounds on the Markovian capacity. 

\subsection{Main Result}
Our main result is the following. 
\begin{theorem}\label{thrm:mainresult}
The capacity for Markovian protocols over BSCs with crossover probability $\eps$ is lower bounded by
\begin{align}
C_{\textrm{Markov}}(\eps) \geq \max\left\{R_0(\eps),\; \sup_{K,M\in\mathbb{N}}R_1(\eps,K,M)\right\}
\end{align}
where $R_0(\eps) \dfn \frac{2}{3}(1-h(\eps))$, 
\begin{align}
R_1(\eps,K,M) \dfn \frac{1-h(\eps)}{1+h(\eps)+\ell(\eps,K,M)},
\end{align}
and $\ell(\eps,K,M)$ is defined in \eqref{eq:elldef}. The following upper bounds for $\ell(\eps,K,M)$ are easily computable and can be used to lower bound $R_1(\eps,K,M)$. The first bound is
\begin{align}
\ell(\eps,K,M)\leq h\left(\minhalf{\eps(2-\eps)}\right)
\end{align}
(where $\minhalf{x}\dfn\min(x,\tfrac{1}{2})$)
and the second, tighter upper bound is $\ell(\eps,K,M)\leq\check{\ell}(\eps)$ where
\begin{align}
\check{\ell}(\eps)= \sum_{k=1}^\infty (\eps(2-\eps))^2(1-\eps(2-\eps))^{k-1}\log(k+1).
\end{align}
\end{theorem}
The rates $R_0(\eps)$ and $R_1(\eps,K,M)$, (with $K=100,M=400$) normalized by the BSC capacity $1-h(\eps)$, are plotted in Fig.~\ref{fig:Rfig}. It can be seen that $R_1(\eps,K,M)$ is superior for $\eps<0.044$, and $R_0(\eps)$ is superior otherwise. Moreover, analyzing $R_1(\eps,K,M)$ for small $\eps$, the following can be shown: 
\begin{corollary} 
For $\eps\to 0$ 
\begin{align}
C_{\mathrm{Markov}}(\eps) =  1-\Theta(h(\eps)).
\end{align}\end{corollary}
In light of the trivial upper bound~\eqref{eq:shannon_capacity}, this rate is order-wise the best possible. Moreover, it is order-wise higher than the lower bound of $1-O(\sqrt{h(\eps)})$ obtained by Kol and Raz \cite{kol2013interactive} for interactive protocols with alternating rounds a non-adaptive transmission schedule.

The remainder of the paper is dedicated to the proof of Theorem~\ref{thrm:mainresult} and its corollary, and is organized follows:
In Subsection~\ref{subsec:scheme1}, we present Scheme \#1, which is a very simple scheme that achieves $R_0(\eps)$. In Subsection~\ref{subsec:scheme2}, we present Scheme \#2 which is more involved and achieves $\sup_{K,M}R_1(\eps,K,M)$ which is larger than $R_0(\eps)$ for any $\eps<0.044$. The analysis of Scheme \#2, including the description of a designated compression protocol, its behavior for large $n$, and numeric evaluation of $R_1(\eps,K,M)$ are given in Section~\ref{sec:analysisscheme2}.

\section{Coding Schemes}
\subsection{Scheme \#1\label{subsec:scheme1}}
We observe that the transmission functions $f_i^A(\cdot),f_i^B(\cdot)$, are binary functions that map a single input bit to a single output bit. We note that there are only four such functions, $\mu_1$, $\mu_2$, $\mu_3$, $\mu_4$ as in the following table:
\begin{align}
\begin{tabular}{|c||*{4}{c|}}
\hline
 &$\mu_1$: & $\mu_2$: & $\mu_3$: & $\mu_4$: \\
$Y$ & $X=Y+ 0$ & $X=Y+ 1$ & $X=0$ & $X=1$ \\
\hline\hline
$0$ & $0$ &  $1$ & $0$ & $1$\\
\hline
$1$ & $1$ &  $0$ & $0$& $1$\\
\hline
\end{tabular}
\end{align}
We observe that $\mu_1$ and $\mu_2$ are \textit{linear}, i.e. $X=Y+ \xi$ and $\xi$ is either $0$ or $1$. $\mu_3$ and $\mu_4$ are constant functions, namely, the output is $0$ or $1$ regardless the input. In the sequel we refer to the locations where $\mu_3$ and $\mu_4$ are used as "stuck positions". 

Having this simple notion stated, we note both and Alice and Bob can compress their own transmission functions using $2n$ bits. We also note that every party, having the transmission functions of its counterpart, can simulate the entire clean transcript. So, we can state the following reliable interaction protocol:
\begin{enumerate}
	\item Alice compresses all her transmission function using $2n$ bits
	\item Alice sends them to Bob using a capacity achieving channel code with rate $1-h(\eps)$. The number of required transmissions from Alice to Bob at this step is $2n/(1-h(\eps)+o(1))$ 
	with error probability $O(1/\text{poly}(n))$.
	\item Bob, having all Alice's transmission functions, can simulate the clean transcript. 
	\item Bob can feed his side of the transcript to Alice, requiring $n$ information bits over a channel with capacity $1-h(\eps)$. So overall $n/(1-h(\eps)+o(1))$ channel uses are needed, 
	with error probability $O(1/\text{poly}(n))$.
\end{enumerate}
So, overall $\tilde{n}=3n/(1-h(\eps)+o(1))$ channel uses are required (with error probability $O(1/\text{poly}(n))$) hence the rate is
\begin{align}
R_0(\eps)= \frac{2n}{3n/(1-h(\eps))}=\frac{2}{3}(1-h(\eps)).
\end{align}

\subsection{Scheme \#2\label{subsec:scheme2}}
The improved achievable rate introduced here is based on running the protocol disregarding the channel errors (as if the channels were clean), followed by several rounds designated to correct the errors.
This scheme is found to be better that to the trivial scheme when the channel noise is low. 
 We start by running the "clean" protocol,
namely Alice and Bob use the Markovian transmission functions on their noisy inputs, $X_i^A = f_i^A(Y_{i-1}^B)$ and  $X_i^B =f_i^B(Y_{i}^A)$ requiring $2n$ channel uses. Then, Alice can describe to Bob the errors of the first round using Slepian-Wolf \cite{slepian1973noiseless} coding protected by a channel code. After this step, the stuck positions are transmitted from side to side using a designated compression algorithm. Finally, the protocol is corrected, using the linearity of the transmission functions in places where they are linear, and reseting at stuck position (as will be elaborated in the sequel).

Let us summarize these steps and give the rate calculation:
\begin{enumerate}
	\item Both parties perform interaction disregarding the channel errors. Overall $2n$ channel uses.
	\item Alice describes Bob the errors that occurred on the channel connecting them (i.e. the channel from Alice to Bob) using Slepian-Wolf coding over a noisy channel. This step requires $ n(h(\eps)+o(1))/(1-h(\eps)+o(1))$ channel uses
	(with error probability $O(1/\text{poly}(n))$). Then Bob feeds the errors back to Alice using simple typical set coding (not Slepian-Wolf). These steps are repeated replacing the roles of Alice and Bob. All in all the channel are used 
	$4n(h(\eps)+o(1))/(1-h(\eps)+o(1))$ times (with error probability $O(1/\text{poly}(n))$). At the end of this step both parties are aware of all channel errors on both sides.		
	\item \label{step3} Bob, knowing all channel errors on both channels divides his interaction functions ,$f_i^A(\cdot)$, into segments that start and end with a channel error (on either channel direction). Then, the first "stuck position" (i.e. $\mu_3$ of $\mu_4$) is conveyed to Alice using the protocol elaborated in Subsection~\ref{sec:stuckpos}. The maximal (i.e. worst case) number of bits used for the description is denoted by $n\ell(\eps)$ and should be conveyed using a capacity achieving channel code requiring $n\ell(\eps)/(1-h(\eps)+o(1))$ channel uses in total. 
	\item \label{step4} Having all this data, Alice can simulate Bob's clean transcript. Assume that from $1\leq i\leq j$ both Alice and Bob have only linear transmission functions. Then, due to the linearity of the transmission at both parties, Bob's clean transcript $\hat{X}_i^B$ can be simulated by canceling the error at both sides:	
	\begin{align}
	\hat{X}_i^B = Y_i^B+ \sum_{l=1}^i Z_{l}^A+ \sum_{l=1}^{i} Z_{i}^B.
	\end{align}
		
	\item \label{step5} Whenever there is a "stuck position" for either party, the processing of previous errors is reset. For example, if Alice receives $Y_i^B=0$, and knows the value of $Z_i^B$ and the fact that
	$f_i^B$ is either $\mu_3$ or $\mu_4$, then $X_i^B=Y_i^B+Z_i^B$, disregarding previous noise values. Note that in non-stuck positions $X_i^B$ is not necessarily equal to $Y_i^B+Z_i^B$. This is because $X_i^B$ is defined as Bob's transmission in the hypothetical noiseless interaction, and not as his transmission in step 1. 

	\item Steps \ref{step3},\ref{step4} and \ref{step5} are repeated by appropriately exchanging the roles of Alice and Bob.
\end{enumerate}
The rate attained by this scheme is therefore
\begin{align}
R_1(\eps,K,M) &=\frac{2n}{2n+n(4h(\eps)+2\ell(\eps,K,M))/(1-h(\eps))}
\\ &=\frac{1-h(\eps)}{1+h(\eps)+\ell(\eps,K,M)}
\end{align}
In the sequel, we will be mostly concerned with the computation of the achievable rate $R_1$. We will also provide a simple lower bound on $R_1$ which is easier to compute, by upper bounding 
 $\check{\ell}(\eps)\geq \ell(\eps)$ (see \eqref{eq:ellcheck}):
\begin{align}
\sup_{K,M}R_1(\eps,K,M) \geq \frac{1-h(\eps)}{1+h(\eps)+\check{\ell}(\eps)}.
\end{align}
The achievable rates are depicted in Fig.~\ref{fig:Rfig}. $R_1$ is computed using $\ell(\eps,K,M)$ and two corresponding upper bounds $\check{\ell}(\eps)$ and a trivial entropy bound that is elaborated in the next section.

It is important to note that for $\eps>0.044$ the description of the errors and stuck positions in scheme \#2 causes it to be less efficients than scheme \#1 as seen in the figure. On the other hand, $\sup_{K,M}R_1(\eps,K,M)$ is better that $R_0(\eps)$ when the channel noise is low and approach $1$ as the $\eps$ go to zero. 
\begin{figure}
	\begin{center}
%
%
\begin{tikzpicture}

\begin{axis}[%
xmin=0,
xmax=0.15,
xtick={0,0.05,0.1,0.15},
xlabel={$\varepsilon$},
xmajorgrids,
ymin=0.3,
ymax=1,
ymajorgrids,
axis background/.style={fill=white},
legend style={legend cell align=left,align=left,draw=white!15!black}
]
\addplot [color=black,solid,line width=1.0pt]
  table[row sep=crcr]{%
0 0.666666666666667\\
0.005	0.666666666666667\\
0.0126315789473684	0.666666666666667\\
0.0202631578947368	0.666666666666667\\
0.0278947368421053	0.666666666666667\\
0.0355263157894737	0.666666666666667\\
0.0431578947368421	0.666666666666667\\
0.0507894736842105	0.666666666666667\\
0.0584210526315789	0.666666666666667\\
0.0660526315789474	0.666666666666667\\
0.0736842105263158	0.666666666666667\\
0.0813157894736842	0.666666666666667\\
0.0889473684210526	0.666666666666667\\
0.096578947368421	0.666666666666667\\
0.104210526315789	0.666666666666667\\
0.111842105263158	0.666666666666667\\
0.119473684210526	0.666666666666667\\
0.127105263157895	0.666666666666667\\
0.134736842105263	0.666666666666667\\
0.142368421052632	0.666666666666667\\
0.15	0.666666666666667\\
};
\addlegendentry{$R_0/(1-h(\varepsilon))$};

\addplot [color=blue,dashed,line width=1.0pt]
  table[row sep=crcr]{%
0 1\\		
0.005	0.932321098471752\\
0.0126315789473684	0.843239333786921\\
0.0202631578947368	0.782401253461246\\
0.0278947368421053	0.737524489979941\\
0.0355263157894737	0.701670508699888\\
0.0431578947368421	0.671699373793775\\
0.0507894736842105	0.645988980357926\\
0.0584210526315789	0.623555934655631\\
0.0660526315789474	0.603735493671447\\
0.0736842105263158	0.58604861656336\\
0.0813157894736842	0.57013614919646\\
0.0889473684210526	0.555721255511242\\
0.096578947368421	0.542585942291193\\
0.104210526315789	0.530555505296965\\
0.111842105263158	0.519487848928373\\
0.119473684210526	0.509265903429939\\
0.127105263157895	0.499792113784815\\
0.134736842105263	0.490984415014525\\
0.142368421052632	0.482773145984384\\
0.15	0.475098706151202\\
};
\addlegendentry{$R_1/(1-h(\varepsilon))$ using $\ell$};

\addplot [color=red,dashdotted,line width=1.0pt]
  table[row sep=crcr]{%
0 1\\
0.005	0.905504486334766\\
0.0126315789473684	0.822981786885182\\
0.0202631578947368	0.76548515004344\\
0.0278947368421053	0.721262662721108\\
0.0355263157894737	0.685508685216725\\
0.0431578947368421	0.655672323636035\\
0.0507894736842105	0.630212803342769\\
0.0584210526315789	0.608121399727921\\
0.0660526315789474	0.588699626533698\\
0.0736842105263158	0.571443220712425\\
0.0813157894736842	0.55597612828637\\
0.0889473684210526	0.542010495264578\\
0.096578947368421	0.529321193769383\\
0.104210526315789	0.517728944791483\\
0.111842105263158	0.507088763081136\\
0.119473684210526	0.497281824035189\\
0.127105263157895	0.488209602205564\\
0.134736842105263	0.479789559564653\\
0.142368421052632	0.471951916413461\\
0.15	0.464637194487272\\
};
\addlegendentry{$R_1/(1-h(\varepsilon))$ using $\check{\ell}$};

\addplot [color=black,dotted,line width=1.0pt]
  table[row sep=crcr]{%
0 1\\
0.005	0.869738067193197\\
0.0126315789473684	0.761691916732077\\
0.0202631578947368	0.689999396292757\\
0.0278947368421053	0.636870031224341\\
0.0355263157894737	0.595181044302708\\
0.0431578947368421	0.561249021245218\\
0.0507894736842105	0.532905084973719\\
0.0584210526315789	0.508762560355629\\
0.0660526315789474	0.487881930148811\\
0.0736842105263158	0.469598318595318\\
0.0813157894736842	0.453424887695311\\
0.0889473684210526	0.438995206617476\\
0.096578947368421	0.426027115407648\\
0.104210526315789	0.414299132440811\\
0.111842105263158	0.403634526383671\\
0.119473684210526	0.393890253735249\\
0.127105263157895	0.384949086677301\\
0.134736842105263	0.376713891727154\\
0.142368421052632	0.369103393859612\\
0.15	0.362048988595575\\
};
\addlegendentry{$R_1/(1-h(\varepsilon))$ using $h$};

\end{axis}
\end{tikzpicture}%
	\end{center}
	\caption{Achievable rates normalized by $1-h(\eps)$. $R_1$ is computed with $K=100,M=400$. \label{fig:Rfig}}
\end{figure}
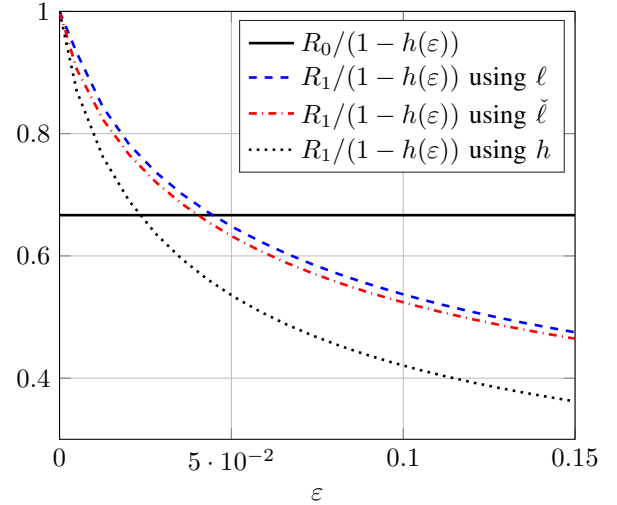

It is of interest to compare this results to \cite{kol2013interactive}. Taking the trivial upper bound $\ell(\eps)\leq h(\minhalf{\eps(2-\eps)})$  given in Subsection~\ref{sec:stuckpos} we can assess the behavior for small $\eps$ by:
\begin{align}
\sup_{K,M}R_1(\eps,K,M)\geq \frac{1-h(\eps)}{1+h(\eps)+h(\eps(2-\eps))}=1-\Theta(h(\eps)),
\end{align}
hence, the capacity for Markovian protocols scales like the Shannon capacity in this limit. 
This should be juxtaposed with the upper bound (for general protocols) of $1-\Omega(\sqrt{h(\eps)})$ given in  \cite{kol2013interactive}. This shows a gap between the capacities of general protocols and Markovian protocols. 
We note that \cite{kol2013interactive} assumes non-adaptive transmission order, which is satisfied by our scheme.

It was shown in \cite{haeupler2014interactive} a higher rate of $1-O(\sqrt{\eps})$ can be achieved for general protocols under adaptive transmission order. This rate is still outperformed by our scheme (for Markovian protocols).

\section{Analysis of Scheme \#2}\label{sec:analysisscheme2}
In this section we analyze the performance of the scheme introduced in \ref{subsec:scheme2}. In particular, we define and analyze a novel compression algorithm designated for the compression of the stuck position.

\subsection{Compression of the Stuck Positions \label{sec:stuckpos}}
We consider the fixed binary sequence  $\bs{\phi}^n=(\phi_1,\ldots,\phi_n)$, $\phi_i\in\{0,1\}$ which describes the "stuck positions" in the original problem. Namely, $\bs{\phi}^n$ describes Bob's "stuck positions", and is equal to $1$ if $f_i^B(\cdot)=\mu_3=0$ or $f_i^B(\cdot)=\mu_4=1$. 
We also consider the i.i.d random sequence $\mybf{z}^n=(z_1,\ldots,z_n)$, $z_i\in\{0,1\}$ with marginal probability $\Pr(z_i=1)=p$, where $p$ is the probability that there is at least one error on the channel from Alice to Bob or vice versa, i.e. $p=1-(1-\eps)^2=\eps(2-\eps)$. 

It is useful to think of the interlaced picture:
\begin{align}
\begin{matrix}
z_1&&z_2&&z_3&&\ldots\\
&\phi_1&&\phi_2&&\phi_3&\ldots
\end{matrix}
\end{align}
The sequence $\mybf{z}^n$ is parsed into segments of the form 
$(1,\mybf{0}^{k-1})$,
$k>0$, where $\mybf{0}^{k-1}$ denotes a row vector of zeros with $k-1$ elements.

We wish to describe the position of the first $\phi_j=1$ in every segment. For example, consider the following interlaced sequence
\begin{align}
\begin{matrix}
\mybf{z}=		&1 &   & 0  &                          & 0, &   & 1 &   & 0, &   & 1, &    & 1 \\
\bs{\phi}=	&  & 0 &	& \text{\textcircled{$1$}} &   & 1 &   & \text{\textcircled{$1$}} &   & 1 &   & \text{\textcircled{$1$}}   &   
\end{matrix}
\end{align}
The parsed segments are separated by commas, and the appearances of the first $\phi=1$ are circled.

First, we note that the total number of the first stuck positions is trivially upper bounded by the number of segments, which is the total number of errors. So, the total number of the first stuck positions is with high probability smaller than $n(p+o(1))$ and can be described via universal compression using less than $n(h(\minhalf{p})+o(1))$ bits. We note that this naive compression method does not use the fact that both sides know the error positions and can take advantage of them in order to improve the compression rate. 

An improved compression algorithm can use the knowledge of the vector $\mybf{z}$ as follows: 
segments of length $k$ are grouped and the empirical distribution of the appearance of the first $1$ is calculated. Then, universal compression is applied for every $k$ based on these distributions. We denote the vector of empirical distribution related to segments of length $k$ by $\bs{\pi}_k=\left\{\pi_{k,l}\right\}_{l=0}^k$. The first $k$ elements of this vector comprise the fraction of these segments that start at some $z_i$, and whose first appearance of $\phi_j=1$ thereafter is at $j=i+l$. $\pi_{k,k}$ is the fraction of the segments that contain no $\phi_j=1$. 

Let $L$ denote the overall length of the stuck positions description (with high probability), normalized by $n$. In the sequel we shall prove that $L$ converges to an asymptotic value $\bar{L}$, which is more easily computable. 

First, we define the empirical distribution $\pi_{k,l}$ (for $0\leq l\leq k$) as the ratio between the counters $N_{k,l}$ and $N_{k}$:
\begin{align}\label{eq:pikl}
\pi_{k,l}=\frac{N_{k,l}}{N_k}.
\end{align}
The counters $N_{k,l}$ are defined as:
\begin{align}
N_{k,l} \dfn \sum_{i=1}^n\boldone_{k,l}(i),
\end{align}
where indicator $\boldone_{k,l}(i)$ for $0\leq l< k$ is one only if and only if
$\mybf{z}_i^{i+k}=(1,\mybf{0}^{k-1},1)$ and 
\begin{align}
\phi_{i+j}=
\begin{cases}
1&\text{for } j=l\\
0&\text{for } 0\leq j<l.
\end{cases}
\end{align}
The indicator $\boldone_{k,k}(i)$ is one only if 
$\mybf{z}_i^{i+k}=(1,\mybf{0}^{k-1},1)$ and $\phi_{i+j}=0$ for $0\leq j< k$.
The denominator of \eqref{eq:pikl} is defined as 
\begin{align}
N_{k} \dfn {\sum_{l=0}^{k} N_{k,l}}=\sum_{i=1}^n\indfunc{\mybf{z}_i^{i+k}=(1,\mybf{0}^{k-1},1)}
\end{align}
where the second equality is by construction. 

Having the counters and the resulting empirical distribution vectors $\bs{\pi}_k$, we can calculate 
the average description length $L$. It is useful to use two schemes, one for $k\leq K_n$ and one for $k>K_n$, with $K_n$ defined in the sequel. For $k<K_n$ we use universal compression which requires for every $k$:
$N_kH\left(\bs{\pi}_k\right)$ bits for the compression where $H(\cdot)$ is the entropy function of a probability vector. Additional bits are also required for the lossless description of the probability vectors $\bs{\pi}_k$ for $k\leq K_n$. We denote this number of bits by $W$.

For $k>K_n$ we describe the location of the first stuck position using the simplifying assumption that $\pi_{k,l}=\frac{1}{k+1}$ (for all $0\leq l\leq k$), shared by both the receiver and transmitter. So, the number of bits for every value of $k$ is $\lceil\log(k+1)\rceil$. All in all, 
the average description length ${L}$ is
\begin{align}
{L} = \frac{1}{n}\left[\sum_{k=1}^{K_n}N_kH\left(\bs{\pi}_k\right)
+W+\sum_{k=K_n+1}^{n}N_k\left\lceil\log(k+1)\right\rceil
\right]
\end{align}

It is useful write $L$ as
\begin{align}\label{eq:Ldef}
{L} = S_1+\frac{W}{n}+S_2
\end{align}	
where 
\begin{align}\label{eq:S1_def}
S_1&\dfn \sum_{k=1}^{K_n}\frac{N_k}{n}H\left(\bs{\pi}_k\right),\\
S_2&\dfn \sum_{k=K_n+1}^{n}\frac{N_k}{n}\left\lceil\log(k+1)\right\rceil.\label{eq:S2_def}
\end{align}

In the sequel we prove that $L$ converges to its asymptotic value by proving that the counters $N_{k,l}$ and $N_k$ converge to their expected values. It is now useful to introduce "spectrum vector" $\{a_m\}$, and write 
$\Expt N_{k,l}$ and $\Expt N_k$ as functions of this vector. Let
\begin{align}\label{eq:amdef}
a_m=\frac{1}{n}\sum_{i=1}^n\indfunc{\phi_{i-m}=1,\bs{\phi}_{i-m+1}^{i-1}=\bs{0}^{m-1},\phi_i=1}
\end{align}
for $m=1,...,n$.
Namely, $a_m$ is the fraction of elements in $\bs{\phi}^n$ which are equal to $1$, and their nearest preceding $1$ in $\mybf{z}^n$ is exactly $m$ time instances earlier. In order to take care of the edge effects we set $\phi_{0}=1$ and $\phi_i=0$ for $i<0$.

Let us now calculate the related expectations:
\begin{align}\label{eq:NkExpt}
\Expt N_k=\sum_{i=1}^{n}\Expt\indfunc{\mybf{z}_i^{i+k}=(1,\mybf{0}^{k-1},1)}=
np^2(1-p)^{k-1}.
\end{align}

And for $0\leq l<k$ 
\begin{align}
&\Expt N_{k,l}
=\sum_{i=1}^n\Expt \left[ \boldone_{k,l}(i)\right]\\
&=\sum_{i=1}^n p^2(1-p)^{k-1}
\indfunc{\bs{\phi}_{i}^{i+l-1}=\mybf{0}^{l-1} \text{ AND } \phi_{i+l}=1 }\\
&\stackrel{(a)}{=}p^2(1-p)^{k-1}n\sum_{m=l+1}^{n} a_m
\end{align}
where $(a)$ follows by counting the number of elements in $\bs{\phi}^n$ that are one, and whose distance to their preceding one is more than $l+1$ (note that $l$ starts at zero).

The calculation of $\Expt N_{k,k}$ is different:
\begin{align}
\Expt N_{k,k}
&= \sum_{i=1}^n\Expt\left[ \boldone_{k,k}(i)\right]\\
&= p^2(1-p)^{k-1}\sum_{i=1}^n\indfunc{\bs{\phi}_i^{i+k-1}=\bs{0}^k}\\
&\stackrel{(a)}{=} p^2(1-p)^{k-1}n\sum_{m=k+1}^{n} a_m(m-k) \label{eq:ENkk}
\end{align}The equality $(a)$ follows by observing that for every $\bs{\phi}$ segment $(1,\bs{0}^{m-1},1)$ there exist $m-k$ placements of a $z$ sequence $(1,\bs{0}^{k-1},1)$ that contain no $\phi=1$. This notion is illustrated below:
\begin{align}
\begin{matrix}
\mybf{z}=&1\overbrace{0\cdots 0}^{k-1}1&\\
\bs{\phi}=&1\underbrace{0\cdots\cdots\cdots 0}_{m-1}1&
\end{matrix}
\end{align}
It is also easy to verify that $\Expt N_k=\sum_{l=0}^k\Expt N_{k,l}$.
%
Let us define the probabilities
\begin{align}
\bar{\pi}_{k,l}&\dfn\frac{\Expt N_{k,l}}{\Expt N_k}\\
&=\begin{cases}
\sum_{m=l+1}^{n} a_m,& \text{ for } 0\leq l< k\\ 
\sum_{m=k+1}^{n} a_m(m-k),&  \text{ for } 0\leq l= k
\end{cases}.\label{eq:pibarkl}
\end{align}
and define $\bar{L}$ based on the definition of $S_1$ in \eqref{eq:S1_def}, replacing $N_k$ with $\Expt N_k$, $\bs{\pi}_{k}$ with $\bar{\bs{\pi}}_{k}$:
\begin{align}
\bar{L} =\frac{1}{n}  \sum_{k=1}^{K_n}np^2(1-p)^{k-1} H\left(\bar{\bs{\pi}}_{k}\right)
\end{align}	
We are now ready to state Theorem~\ref{thrm:Lconverge}.
\subsection{Convergence of the Compression Rate}
\begin{theorem}\label{thrm:Lconverge}
For $K_n=\frac{\beta\ln n}{-\ln(1-p)}$ 
for every $\eps>0$ 
\begin{align}
\lim_{n\to \infty}\Pr\left(L>\bar{L}+\eps\right)=0.
\end{align}

\end{theorem}
Recalling \eqref{eq:Ldef}, $L$ is composed of three elements : $S_1$, $W/n$ and $S_2$. Proving that $W/n$ and $S_2$ converge to zero is simple and is deferred to the end of this subsection. Analyzing the convergence of $S_1$ is more involved and is now handled. The proof is based on two elements: the convergence of the counters $N_{k,l}$ to their expected value (Lemma~\ref{lemma:concentration}), and the smoothness of the entropy function (Lemma~\ref{lemma:smoothH}). Let us start by giving the lemmas and then use them to prove the theorem. 

\begin{lemma}\label{lemma:concentration}
The following inequalities hold any $t\geq 0$:
\begin{align}
&\Pr\left(\left| {N}_{k,l}-\Expt {N}_{k,l} \right|\geq t\right)
\leq 2\exp\left(-\tfrac{t^2}{16n}\right) \label{eq:lemma1:eq1}\\
&\Pr\left( \left|{N}_{k}-\Expt {N}_{k}\right| \geq t\right)
\leq 2\exp\left(-\tfrac{t^2}{16n}\right)\label{eq:lemma1:eq3}\\
&\Pr\left(\sum_{k=K+1}^{n} {N}_{k}-\Expt \sum_{k=K+1}^{n}{N}_{k} \geq t\right)
\leq \exp\left(-\tfrac{t^2}{4n}\right).\label{eq:lemma1:eq5}\\
\end{align}
\end{lemma}
\begin{proof}
The proof is based on a straightforward application of the bounded difference inequality. We start by citing the inequality:
\begin{theorem}[\textit{Bounded difference inequality} {\cite[Theorem 3.18]{VanHandel}} ]
	
	 Let $\mybf{x}^n$ be a random independent series, and $f(\mybf{x}^n)$ a scalar function, then for any $t\geq 0$ the following hold:
	\begin{align}	
	\Pr(f(\mybf{x}^n)-\Expt &f(\mybf{x}^n)\geq t ) \\&\leq 
	\exp\left(-\tfrac{t^2}{4\norm{\sum_{i=0}^{n}|D_i^- f|^2}_{\infty}}\right)\label{bdi1}\\
	\Pr(f(\mybf{x}^n)-\Expt& f(\mybf{x}^n)\leq  -t) \\&\leq  
\exp\left(-\tfrac{t^2}{4\norm{\sum_{i=0}^{n}|D_i^+ f|^2}_{\infty}}\right).\label{bdi2}
	\end{align}
where 	
\begin{align}
D_i^{-} f \dfn f(\mybf{x}^n)-\inf_{x} f(\mybf{x}^{i-1},x,\mybf{x}_{i+1}^n)
\end{align}
and
\begin{align}
D_i^{+} f \dfn \sup_{x} f(\mybf{x}^{i-1},x,\mybf{x}_{i+1}^n) - f(\mybf{x}^n).
\end{align}
\end{theorem}

We shall use the theorem by setting $f(\mybf{x}^n)={N}_{k,l}$ where $\mybf{x}^n$ is the noise series $\mybf{z}^n$ (i.i.d $\text{Ber}(p)$). Using this, the elements of $\mybf{z}^n$ determine the error segments in which the counters ${N}_{k,l}$ are calculated. We observe that changing a single element of $\mybf{z}^n$ can leave the number of segments unchanged or change them by at most two. The maximal change is achieved in the following situation:
\begin{align}
\mybf{z}^n=(\ldots,1,\mybf{0}^{k-1},x_i,\mybf{0}^{k-1},1\ldots)
\end{align}
in which changing $x_i$ from zero to one (respectively from one to zero) will increase (respectively decrease) the number of segments by two. Since changing the number of segments by two will change the counter $N_{k,l}$ by at most two we can conclude that
$D_i^+f\leq 2$ and $D_i^-f\leq 2$, and  
\begin{align}
\norm{\sum_{i=0}^{n}|D_i^- f|^2}_{\infty}\leq n\left(2\right)^2=4n
\end{align}
and similarly $\norm{\sum_{i=0}^{n}|D_i^+ f|^2}_{\infty}$ can be bounded by the same value. Using this and \eqref{bdi1} and \eqref{bdi2} we obtain \eqref{eq:lemma1:eq1}. 
Note that the same bound and the same argument also holds for the total number of segments $N_k$ expressed in \eqref{eq:lemma1:eq3}.

Finally, \eqref{eq:lemma1:eq5} follows from the fact that changing $x_i$ from one to zero will create at most one new segment with $k\geq K+1$. So $D_i^+\leq 1$ and $\norm{\sum_{i=0}^{n}|D_i^+ f|^2}_{\infty}\leq n$.
\end{proof}

\begin{lemma}[{\cite[Lemma~2.7]{CsiszarKorner}}]\label{lemma:smoothH}
	If $d(P,Q)=\Theta\leq \frac{1}{2}$ then 
	\begin{align}
	|H(P)-H(Q)|\leq -\Theta\log\frac{\Theta}{|\mathcal{X}|}
	\end{align}
	where $d_{\mathrm{TV}}(P,Q)$ is the \textit{total variation distance} between the distributions $P$ and $Q$ on $\mathcal{X}$:
	\begin{align}
	d_{\mathrm{TV}}(P,Q)\dfn \sum_{x\in \mathcal{X}}|P(x)-Q(x)|.
	\end{align} 
\end{lemma}

Having these two lemmas at hand, we are ready to prove the Theorem~\ref{thrm:Lconverge}.
\begin{proof}[Proof of Theorem~\ref{thrm:Lconverge}]
We start by proving that $S_1$ is asymptotically upper bounded by $\bar{L}$. We first upper bound $|H\left(\bs{\pi}_k\right)-H\left(\bar{\bs{\pi}}_k\right)|$ by upper bounding the variation distance $\Theta=|\bs{\pi}_k-\bar{\bs{\pi}}_k|$. Using \eqref{eq:lemma1:eq1} and \eqref{eq:lemma1:eq3} and the union bound it follows that
\begin{align}
\Pr\left(\frac{N_{k,l}}{N_k}\geq \frac{\Expt N_{k,l}+t}{\Expt N_k-t}\right)\leq 
4\exp\left(-\tfrac{t^2}{16n}\right).
\end{align}
where $t<\Expt N_k$. 
Using \eqref{eq:pikl}, \eqref{eq:pibarkl} and the fact that $\Expt N_k=np^2(1-p)^{k-1}$ we can also write 
\begin{align}
\Pr&\left(\frac{N_{k,l}}{N_k}\geq \frac{\Expt N_{k,l}+t}{\Expt N_k-t}\right)=
\Pr\left(\pi_{k,l}\geq \tfrac{np^2(1-p)^{k-1}\bar{\pi}_{k,l}+t}{np^2(1-p)^{k-1}-t}\right)\\
=&\Pr\left(\pi_{k,l}-\bar{\pi}_{k,l}\geq \tfrac{t(\bar{\pi}_{k,l}+1)}{np^2(1-p)^{k-1}-t}\right)
\end{align}
and hence
\begin{align}\label{eq:piklupper}
\Pr\left(\pi_{k,l}-\bar{\pi}_{k,l}\geq \tfrac{t(\bar{\pi}_{k,l}+1)}{np^2(1-p)^{k-1}-t}\right)\leq 4\exp\left(-\tfrac{t^2}{16n}\right). 
\end{align}
Similarly
\begin{align}
&\Pr\left(\frac{N_{k,l}}{N_k}\leq \frac{\Expt N_{k,l}-t}{\Expt N_k+t}\right)\\
&=\Pr\left(\pi_{k,l}-\bar{\pi}_{k,l}\leq -\tfrac{t(\bar{\pi}_{k,l}+1)}{np^2(1-p)^{k-1}+t}\right)\\
&\leq 4\exp\left(-\tfrac{t^2}{16n}\right).\label{eq:piklower}
\end{align}
Combining \eqref{eq:piklupper} and \eqref{eq:piklower} taking into account that \eqref{eq:piklupper} is tighter, we obtain
\begin{align}
\Pr\left(\left|\pi_{k,l}-\bar{\pi}_{k,l}\right|\geq \tfrac{t(\bar{\pi}_{k,l}+1)}{np^2(1-p)^{k-1}-t}\right)
\leq 8\exp\left(-\tfrac{t^2}{16n}\right).
\end{align}
Summing up for $l=1,\ldots,k$ and using the fact that $\sum_{i=0}^k\bar{\pi}_{k,l}=1$ we get the following inequality for the variation distance $\Theta=d_{\mathrm{TV}}(\bs{\pi}_{k},\bar{\bs{\pi}}_{k})$:
\begin{align}
\Pr\left(\Theta\geq \tfrac{t(k+2)}{np^2(1-p)^{k-1}-t}\right)
\leq 8(k+1)\exp\left(-\tfrac{t^2}{16n}\right).
\end{align}
Now, we can set $t=n^\alpha$ with $\alpha\in(\frac{1}{2},1)$ and obtain
\begin{align}
\Pr\bigg(\Theta\geq \tfrac{k+2}{n^{1-\alpha}p^2(1-p)^{k-1}+1}\bigg)&\leq 8(k+1)\exp\left(-\tfrac{n^{2\alpha-1}}{16}\right).
\end{align}
Finally, Lemma~\ref{lemma:smoothH} implies that
\begin{align}\label{eq:Hconvergence}
\Pr\bigg(| H({\bs{\pi}}_{k})- H(\bar{\bs{\pi}}_{k})|&\geq \eps_k
\bigg)\\
&\leq 8(k+1)\exp\left(-\tfrac{n^{2\alpha-1}}{16}\right).
\end{align}
where
\begin{align}
\eps_k\dfn -\tfrac{k+2}{n^{1-\alpha}p^2(1-p)^{k-1}-1}\log\left(\tfrac{(k+2)/(k+1)}{n^{1-\alpha}p^2(1-p)^{k-1}-1}\right).
\end{align}
Clearly, for any fixed $k$ we have that $\eps_k\stackrel{n\to \infty}{\longrightarrow} 0$. 
Trivially, the upper side of the bound in \eqref{eq:Hconvergence} also holds:
\begin{align}\label{eq:Hkupperbound}
\Pr( H({\bs{\pi}}_{k})- H(\bar{\bs{\pi}}_{k})\geq \eps_k) \leq 8(k+1)\exp(-\tfrac{n^{2\alpha-1}}{16}).
\end{align}

Let us now bound the summands of $S_1$. Using \eqref{eq:Hkupperbound}, \eqref{eq:lemma1:eq1} and \eqref{eq:NkExpt} and the union bound implies that
\begin{align}
&\Pr\left( \tfrac{N_k}{n} H\left({\bs{\pi}}_{k}\right) \geq (p^2(1-p)^{k-1}+n^{\alpha-1})\left(H\left(\bar{\bs{\pi}}_{k}\right)+\eps_k\right)
\right)\\
&\leq (8k+10)\exp\left(-\tfrac{n^{2\alpha-1}}{16}\right).\label{eq:NknHpibound}
\end{align}

Rearranging \eqref{eq:NknHpibound} we obtain
\begin{align}
\Pr\bigg( \frac{N_k}{n} &H\left({\bs{\pi}}_{k}\right) 
-p^2(1-p)^{k-1}H\left(\bar{\bs{\pi}}_{k}\right)\\
\geq& H(\bar{\bs{\pi}}_{k})n^{\alpha-1}+
 (p^2(1-p)^{k-1}+n^{\alpha-1})\eps_k \bigg)\\
&\leq (8k+10)\exp\left(-\tfrac{n^{2\alpha-1}}{16}\right).
\end{align}
Summing up for $k=1,\ldots,K_n$ and noticing that the following bounds hold for $1\leq k\leq K_n$:
\begin{align}
&p^2(1-p)^{k-1}\leq p^2\\
& H(\bar{\bs{\pi}}_k)\leq \log(K_n+1) \\
& \eps_k\leq \eps_{K_n}
\end{align}
we obtain
\begin{align}\label{eq:S1convergence}
\Pr\left( S_1-\bar{L}\geq {\epsilon_1} \right)\leq { \delta_1}.
\end{align}
where
\begin{align}
\epsilon_1\dfn K_n\log(K_n+1)n^{\alpha-1}+(p^2+n^{\alpha-1})K_n\eps_{K_n}
\end{align}
and
\begin{align}
\delta_1 \dfn K_n(8K_n+10)\exp\left(-\tfrac{n^{2\alpha-1}}{16}\right).
\end{align}

Setting 
\begin{align}\label{eq:Kndef}
K_n=\frac{\beta\ln n}{-\ln(1-p)}
\end{align}
with $\beta\in (0,1-\alpha)$ yields
\begin{align}\label{eq:Knsetting}
(1-p)^{K_n}=n^{-\beta}
\end{align}
which assures that $\eps_{K_n}\stackrel{n\to \infty}{\longrightarrow} 0$  and also 
$\epsilon_1 \stackrel{n\to \infty}{\longrightarrow} =0$ and $\delta_1 \stackrel{n\to \infty}{\longrightarrow} 0$. 

Let us now prove that $S_2$ converges to zero in probability. For $k>K_n$ we describe the location of the first stuck position using $\left\lceil\log(k+1)\right\rceil$ bits for every value of $k$. Therefore
\begin{align}
S_2&= \sum_{k=K_n+1}^{n}\frac{N_k}{n}\left\lceil\log(k+1)\right\rceil\\
&\leq \sum_{k=K_n+1}^{n}\frac{N_k}{n}\left(\log(k+1)+1\right)\\
&\leq \left(\sum_{k=K_n+1}^{n}\frac{N_k}{n}\right)\left(\log(n+1)+1\right) \label{eq:NknHbound}
\end{align}
Recalling \eqref{eq:lemma1:eq5} and using only the upper side of the bound 
\begin{align}
\Pr\bigg(\sum_{k=K_n+1}^{n} \frac{{N}_{k}}{n}\geq  \sum_{k=K+1}^{n}p^2(1-p)^{k-1} &+ \frac{t}{n}\bigg)\\
&\leq 2\exp\left(-\tfrac{t^2}{4n}\right). 
\end{align}
and noticing that 
\begin{align}
\sum_{k=K+1}^{n}p^2(1-p)^{k-1}\leq \sum_{K+1}^{\infty}p^2(1-p)^{k-1} =p(1-p)^{K_n}
\end{align}
 we have
\begin{align}
\Pr\left(\sum_{k=K_n+1}^{n} \frac{{N}_{k}}{n}\geq  p(1-p)^{K_n} + \frac{t}{n}\right)
\leq 2\exp\left(-\tfrac{t^2}{4n}\right)
\end{align}
setting as before $t=n^{\alpha}$ with $\alpha\in(\frac{1}{2},1)$ and recalling \eqref{eq:Knsetting} we obtain
\begin{align}
\Pr\left(\sum_{k=K_n+1}^{n} \frac{{N}_{k}}{n}\geq  pn^{-\beta} + n^{\alpha-1}\right)
\leq 2\exp\left(-\tfrac{n^{2\alpha-1}}{4}\right)
\end{align}

Now, we use \eqref{eq:NknHbound} and further loosen the bound, obtaining
\begin{align}
\Pr\bigg(S_2\geq  (pn^{-\beta} &+ n^{\alpha-1})
\left(\log(n+1)+1\right)
\bigg)\\
&\leq 2\exp(-\tfrac{n^{2\alpha-1}}{4}).
\end{align}
Defining 
\begin{align}
{\epsilon_2}\dfn{\left(pn^{-\beta} + n^{\alpha-1}\right)\log(n+1)}
\end{align}
and 
\begin{align}
\delta_2\dfn{2\exp\left(-\tfrac{n^{2\alpha-1}}{4}\right)}.
\end{align}
we obtain 
\begin{align}\label{eq:S2convergence}
\Pr\left(S_2
\geq  \epsilon_2\right)
\leq \delta_2
\end{align}
where clearly $\epsilon_2 \stackrel{n\to \infty}{\longrightarrow} =0$ and $\delta_2 \stackrel{n\to \infty}{\longrightarrow} 0$. 

Lastly, we can show that $W/n$ converges to zero in probability by representing the values in $\bs{\pi}_k$ for  $k=1,\ldots,K_n$ using $\log(K_n)$ bits each. There are overall $\sum_{k=1}^{K_n}(k+1)=K_n(K_n+3)/2$ such elements so, the total number of required bits is $W = K_n(K_n+3)\log(K_n+1)$. Setting $K_n$ as in \eqref{eq:Kndef} clearly yields $W/n\stackrel{n\to \infty}{\longrightarrow} =0$. Combining this, \eqref{eq:S1convergence}, \eqref{eq:S2convergence} and applying the union concludes the proof. 
\end{proof}

\subsection{Numerical Evaluations of the Compression Rate}
In the previous subsection, we proved that $L$ is asymptotically upper bounded by $\bar{L}$. However, $\bar{L}$ is a function of the spectrum vector $\mybf{a}$. Therefore, an upper bound for $\bar{L}$ should be related to the maximization of $\bar{L}$ w.r.t $\mybf{a}$. In this subsection we explicitly write this (convex) optimization problem, and provides some numeric evaluations. We note that $\mybf{a}$ is a vector of length $n\to \infty$. Since our optimization tools are limited to vectors with finite dimension, we limit the size of $\mybf{a}$, and bound the residue inflicted by this process. 

We start by recalling \eqref{eq:Ldef}: ${L} = S_1+\frac{W}{n}+S_2$,
however, in contrast to the definitions in \eqref{eq:S1_def} and \eqref{eq:S2_def}, we define $S_1$ and $S_2$ with $K$ that is a fixed number, and not an increasing function in $n$. Namely
\begin{align}
S_1&= \sum_{k=1}^{K}\frac{N_k}{n}H\left(\bs{\pi}_k\right),\\
S_2&= \sum_{k=K+1}^{n}\frac{N_k}{n}\left\lceil\log(k+1)\right\rceil.
\end{align}
Having a fixed $K$, the number of bits required for the description of the universal codebooks, $W$ can be trivially upper bounded by $K^2\log(K+1)$ thus clearly $\frac{W}{n}\stackrel{n\to\infty}{\longrightarrow}0$. 
In the previous subsection, we showed that $L$ is asymptotically upper bounded by $\bar{L}$, for $K_n$ 
defined in \eqref{eq:Kndef}. It is possible to show by steps similar to the ones used in the previous subsection that for a fixed $K$, $S_1$ and $S_2$ are asymptotically upper bounded by the following terms respectively
\begin{align}\label{eq:S1pK}
\bar{S}_1(p,K) &\dfn\sum_{k=1}^{K}p^2(1-p)^{k-1}H(\bar{\bs{\pi}}_k),\\
\bar{S}_2(p,K)&\dfn\sum_{k=K+1}^{n}p^2(1-p)^{k-1}\left\lceil\log(k+1)\right\rceil.
\end{align}
Thus the total description length can be written as
\begin{align}\label{eqLpK}
L(p,K)=\bar{S}_1(p,K)+\bar{S}_2(p,K)
\end{align}

We first note that we can trivially upper bound all $H(\bar{\bs{\pi}}_k)$ by $\log(k+1)$ yielding the following bound
\begin{align}\label{eq:Lcheck}
L(p,K)\leq {\check{L}}(p)\dfn\sum_{k=1}^\infty p^2(1-p)^{k-1}\log(k+1).
\end{align}

Let us now write $\bar{S}_1(p,L)$ as a convex optimization problem in $\mybf{a}$, and numerically evaluate its optimum. We recall that $\bar{\bs{\pi}}_k$ can be written in terms of $\mybf{a}$ as in \eqref{eq:pibarkl}. This relation can be stated in using matrix/vector notation by introducing the set of matrices $B_k$ with sizes $(k+1)\times n$ with the following element. For $1\leq i\leq k$ 
\begin{align}
\left[ B_k\right]_{i,j}=
\begin{cases}
1 & \text{for } j\geq i\\
0& \text{otherwise}
\end{cases}
\end{align}
and for $i=k+1$
\begin{align}
\left[ B_k\right]_{k+1,j}=
\begin{cases}
j-k & \text{for } j\geq k \\
0& \text{otherwise}
\end{cases}
\end{align}
The matrix $B_k$ can also be written as follows:
\begin{align}
B_k=
\kbordermatrix{&1&2&3&\cdots&k&k+1&k+2\\
1&1&1&1&\cdots&1&1&1&\cdots\\
2&0&1&1&\cdots&1&1&1&\cdots\\
3&0&0&1&\cdots&1&1&1&\cdots\\
\vdots&&&&\ddots\\
k&0&0&0&\cdots&1&1&1&\cdots\\
k+1&0&0&0&\cdots&0&1&2&\cdots
}
\end{align}

Recalling the definition of $\{a_m\}$ in \eqref{eq:amdef}, and taking into account that all the sequences of all lengths $m=1,\ldots,n$ construct the original sequence $\bs{\phi}^n$ (whose length is $n$) gives $\sum_{m=1}^{n}mna_m=n$ hence $\sum_{m=1}^{n}ma_m=1$. Also note that $a_m\geq 0$ for all $m=1,\ldots,n$. 

So, an upper bound for $\bar{S}_1(p,K)$ denoted by $\tilde{S}_1(p,K)$ can be computed as follows:
\begin{align}\label{eq:Ltildeopt}
\tilde{S}_1(p,K) =\max_{\mybf{a}}&\sum_{k=1}^{K}p^2(1-p)^{k-1}H\left(B_k\mybf{a}\right)\\
\text{{s.t.}  }& a_i\geq 0\quad \forall i, \quad \sum ia_i=1
\end{align}
We note that the constraints are convex and that the function to be maximized is the sum of the composition of convex function ($H(\cdot)$) with linear functions, hence is also convex.

A more convenient parameterization is obtained by normalizing $\mybf{a}$ to be a probability vector. To that end, $B_k$ should be replaced with $C_k$ as follows
\begin{align}\label{eq:Ckdef}
C_k=
\kbordermatrix{&1&2&3&\cdots&k&k+1&k+2\\
	1&1&\frac{1}{2}&\frac{1}{3}&\cdots&\frac{1}{k}&\frac{1}{k+1}&\frac{1}{k+2}&\cdots\\
	2&0&\frac{1}{2}&\frac{1}{3}&\cdots&\frac{1}{k}&\frac{1}{k+1}&\frac{1}{k+2}&\cdots\\
	3&0&0&\frac{1}{3}&\cdots&\frac{1}{k}&\frac{1}{k+1}&\frac{1}{k+2}&\cdots\\
	\vdots&&&&\ddots\\
	k&0&0&0&\cdots&\frac{1}{k}&\frac{1}{k+1}&\frac{1}{k+2}&\cdots\\
	k+1&0&0&0&\cdots&0&\frac{1}{k+1}&\frac{2}{k+2}&\cdots
}.
\end{align}
This yields the following optimization problem 
\begin{align}\label{eq:Ltildeoptprob}
\tilde{S}_1(p,K) =\max_{\mybf{a}}&\sum_{k=1}^{n}p^2(1-p)^{k-1}H\left(C_k\mybf{a}\right)\\
\text{{s.t.}  }& a_i\geq 0\quad \forall i, \quad \sum a_i=1
\end{align}

We would like evaluate $\tilde{S}_1(p,K)$ by numerically optimizing \eqref{eq:Ltildeoptprob}. It is clear that the length of the vector $\mybf{a}$ should be limited to some fixed value (denoted by $M$). 
We denote the reduced size vector by $\tilde{\mybf{a}}$ and derive it from $\mybf{a}$ by:
\begin{align}
\tilde{a}_i=
\begin{cases}
a_i,\quad &\text{ for } i=1,\ldots,M-1\\
\sum_{j=M}^na_j &\text{ for } i=M
\end{cases}
\end{align}
We also define $\tilde{C}_k$ by cutting only the first $M$ columns of $C_k$.
Noticing the definition of $C_k$ in \eqref{eq:Ckdef} and the fact that both $\mybf{a}$ and $\tilde{\mybf{a}}$ are probability vector gives the following bounds 
\begin{align}
\left[\tilde{C}_k\tilde{\mybf{a}}-C_k{\mybf{a}}\right]_i\leq \frac{\tilde{a}_{M}}{M}
\end{align}
for $i\in[1,k]$ and 
\begin{align}
\left|\left[\tilde{C}_k\tilde{\mybf{a}}-C_k{\mybf{a}}\right]_{k+1}\right|=\tilde{a}_M\left|\frac{M-k}{M}-1 \right|
=\frac{k\tilde{a}_{M}}{M}
\end{align}
Therefore, the variation distance is bounded by
\begin{align}
d_{\mathrm{TV}}(\tilde{C}_k\tilde{\mybf{a}},C_k{\mybf{a}})\leq \frac{2k\tilde{a}_{M}}{M}
\end{align}
Lemma~\ref{lemma:smoothH} requires that $d_{\mathrm{TV}}(\tilde{C}_k\tilde{\mybf{a}},C_k{\mybf{a}})\leq\frac{1}{2}$, so in order to comply we set $M=4K$ and obtain the bound:
\begin{align}
H(C_k{\mybf{a}}) <H(\tilde{C}_k\tilde{\mybf{a}}) -\frac{2k\tilde{a}_{M}}{M}\log\frac{2k\tilde{a}_{M}}{(k+1)M}.
\end{align}
Finally, the following finite-dimensional convex optimization problem provides a computable upper bound for $\check{S}_1(p,K,M) \geq \tilde{S}_1(p,K)$ that holds for any $n$ large enough:
\begin{align}
&\check{S}_1(p,K,M) \dfn \label{eq:tildeL0} \\
&\max_{\mybf{\tilde{a}}\in\mathbb{R}^M}\sum_{k=1}^{K}\left[p^2(1-p)^{k-1}H\left(\tilde{C}_k\tilde{\mybf{a}}\right)
-\frac{2k\tilde{a}_{M}}{M}\log\frac{2k\tilde{a}_{M}}{(k+1)M}\right]\\
&\text{{s.t.}  } \quad \tilde{a}_i\geq 0, \quad \sum_{i=1}^M \tilde{a}_i=1
\end{align}

and lastly
\begin{align}
L(p,K,M)&\leq \check{S}_1(p,K,M)+\\
&\sum_{k=K+1}^{n}p^2(1-p)^{k-1}\log(k+1).
\end{align}
We evaluated $\check{L}(p)$ and ${L}(p,K,M)$ for $K=100$ and $M=400$. The results are depicted in Fig.~\ref{fig:Lfig} including the trivial bound $h(p)$.
\begin{figure}
	\begin{center}
%
%
\begin{tikzpicture}

\begin{axis}[%
xmin=0.05,
xmax=0.5,
xlabel={$p$},
xmajorgrids,
ymin=0.1,
ymax=1,
ymajorgrids,
axis background/.style={fill=white},
legend style={at={(0.97,0.03)},anchor=south east,legend cell align=left,align=left,draw=white!15!black}
]
\addplot [color=black,dashdotted,line width=1.0pt]
  table[row sep=crcr]{%
0.05	0.286396957115956\\
0.1	0.468995593589281\\
0.15	0.6098403047164\\
0.2	0.721928094887362\\
0.25	0.811278124459133\\
0.3	0.881290899230693\\
0.35	0.934068055375491\\
0.4	0.970950594454669\\
0.45	0.992774453987808\\
0.5	1\\
};
\addlegendentry{$h(p)$};

\addplot [color=black,dashed,line width=1.0pt]
  table[row sep=crcr]{%
0.05	0.190081535391065\\
0.1	0.298903843234591\\
0.15	0.38282312949711\\
0.2	0.45245697324989\\
0.25	0.512479430310913\\
0.3	0.565471469964669\\
0.35	0.613043289425909\\
0.4	0.656280982722924\\
0.45	0.695958326542502\\
0.5	0.732649482117484\\
};
\addlegendentry{$\check{L}(p)$};

\addplot [color=black,solid,line width=1.0pt]
  table[row sep=crcr]{%
0.05	0.164390472136274\\
0.1	0.261611966338142\\
0.15	0.340238225671621\\
0.2	0.407544353380068\\
0.25	0.467083770389027\\
0.3	0.520856229656279\\
0.35	0.570125735149548\\
0.4	0.615756108469975\\
0.45	0.658370527375441\\
0.5	0.698438186764258\\
};
\addlegendentry{$\tilde{L}(p)$};

\end{axis}
\end{tikzpicture}%
	\end{center}
\caption{ $h(p)$, $\tilde{L}(p)$ and $\check{L}(p)$ as a function of $p$.\label{fig:Lfig}}
\end{figure}
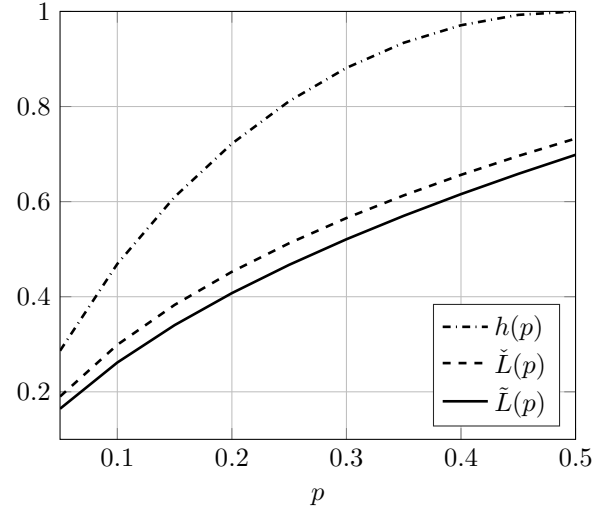

We are only left with relating ${L}(p,K,M)$ to $\ell(\eps,K,M)$.
We note that $p$ corresponds to the event of one of more errors on the channel between Alice and Bob and vice versa. So, $p=1-(1-\eps)^2=\eps(2-\eps)$ and
\begin{align}\label{eq:elldef}
\ell(\eps,K,M)={L}(\eps(2-\eps),K,M),
\end{align}
where $\tilde{L}(\cdot)$ is given in~\eqref{eq:tildeL0}. A simpler upper bound can be obtained using \eqref{eq:Lcheck}, 
\begin{align}\label{eq:ellcheck}
\sup_{K,M}\ell(\eps,K,M)\leq \check{\ell}(\eps) \dfn \check{L}(\eps(2-\eps)).
\end{align}

\bibliographystyle{IEEEbib}
\bibliography{bibtex_references}

\end{document}